\documentclass[12pt]{amsart}

\usepackage{amsmath,amsthm,amssymb,dsfont,enumerate,tikz,hyperref}
\usepackage{setspace}
\usepackage[margin=1.25in]{geometry}

\newtheorem{theorem}{Theorem}[section]
\newtheorem{proposition}[theorem]{Proposition}

\theoremstyle{definition}

\newtheorem{remark}[theorem]{Remark}

\numberwithin{equation}{section}
\numberwithin{figure}{section}

\renewcommand{\subset}{\subseteq}
\renewcommand{\hat}{\widehat}
\renewcommand{\epsilon}{\varepsilon}

\def\<{\langle}
\def\>{\rangle}
\def\({\Big(}
\def\){\Big)}
\def\C{\mathbb{C}}

\def\R{\mathbb{R}}

\def\T{\mathbb{T}}
\def\Z{\mathbb{Z}}

\title{Elementary $L^\infty$ error estimates for super-resolution de-noising}

\author{Weilin Li}



\address{Norbert Wiener Center, Department of Mathematics, University of Maryland, College Park, MD 20742, USA}
\email{wl298@math.umd.edu}

\begin{document}
	
\maketitle

% % % % % % % % % % % % % % % % % % % %
\begin{abstract}
This paper studies the problem of recovering a discrete complex measure on the torus from a finite number of corrupted Fourier samples. We assume the support of the unknown discrete measure satisfies a minimum separation condition and we use convex regularization methods to recover approximations of the original measure. We focus on two well-known convex regularization methods, and for both, we establish an error estimate that bounds the smoothed-out error in terms of the target resolution and noise level. Our $L^\infty$ approximation rate is entirely new for one of the methods, and improves upon a previously established $L^1$ estimate for the other. We provide a unified analysis and an elementary proof of the theorem.
\end{abstract}

%\textbf{Keywords:} super-resolution, convex regularization, total variation
%
%\textbf{2010 Mathematics Subject Classification:} 42A15, 90C25, 94A12
%
%\textbf{Address:} Department of Mathematics, University of Maryland, College Park

% % % % % % % % % % % % % % % % % % % %
\section{Introduction}

% % % % % % % % % % % % % % % % % % % %
\subsection{Overview and Contributions}
Super-resolution techniques are concerned with the recovery of high resolution features from coarse observations, and can be employed to capture information beyond the inherent resolution limit of the measurement system. Applications of super-resolution include microscopy \cite{lindberg2012mathematical}, astronomy \cite{puschmann2005super}, neuroscience \cite{rieke1999spikes}, medical imaging \cite{greenspan2009super}, and geophysics \cite{khaidukov2004diffraction}. 

While there are numerous empirical results on super-resolution, the theory is still in its infancy. Cand\`{e}s and Fernandez-Granda \cite{candes2014towards} introduced a super-resolution model where the unknown information is a discrete and periodic measure whose support set satisfies a minimum separation condition. They proved that such a measure can be uniquely recovered from a finite number of consecutive Fourier samples by solving a convex minimization problem. Several other papers \cite{de2012exact, tang2013compressed, aubel2015theory, duval2015exact} have addressed variations of this model, but it is also possible to study the super-resolution of non-discrete measures \cite{benedetto2016super}. 

The literature also focused on the closely related model where the low resolution data is corrupted by additive noise. This situation is important because in applications of the theory, there might be noise due to measurement error, data corruption, or quantization. Under the same minimum separation framework, several papers \cite{candes2013super, fernandez2013support, bhaskar2013atomic, tang2015near, azais2015spike, duval2015exact} used one of two important convex regularization methods, which we shall call Problem (\ref{SR1}) and (\ref{SR2}), to obtain approximations of the original measure. 

We also adopt the minimum separation model, but unlike the aforementioned papers, we address both convex regularization methods in a \emph{unified} manner. To do this, we show that both methods produce measures that satisfy two (fairly weak) inequalities. We prove that any measure enjoying these properties approximates the unknown measure in a natural sense, and in particular, we can bound the error in terms of the target resolution and noise level. We prove this result using well-known techniques, but we combine the pieces in a different and more efficient manner, resulting in a \emph{significantly simpler} argument. Our $L^\infty$ error estimate is entirely new for Problem (\ref{SR2}) and improves upon a previously established $L^1$ result for Problem (\ref{SR1}), derived by Cand\`{e}s and Fernandez-Granda \cite{candes2013super}.

% % % % % % % % % % % % % % % % % % % %
\subsection{Background}

We first introduce some notation and discuss prior work on the noiseless case. While our results generalize to higher dimensions, for simplicity, we focus on the one-dimensional case. 

Let $M(\T)$ be the space of complex Radon measures on the torus group $\T=\R/\Z$. For $\mu\in M(\T)$, let $|\mu|$ be its variation, $\hat\mu$ be its Fourier transform, and $\|\mu\|$ be its total variation.
%,
%\[
%\|\mu\|
%=|\mu|(\T)
%=\int_{\T}\ d|\mu|.
%\]
%, defined as
%\[
%\hat\mu(m)=\int_{\T} e^{-2\pi imx}\ d\mu(x). 
%\]
For any integer $M>0$, let $\Lambda_M=\{-M,-M+1,\dots,M\}$ and define the projection $P_M\colon M(\T)\to C(\T;\Lambda_M)$ by
\[
P_M\mu(x)
=\sum_{m=-M}^M \hat\mu(m) e^{2\pi imx}, 
\]
where $C(\T;\Lambda_M)$ is the space of trigonometric polynomials of degree $M$, i.e.,
\[
C(\T;\Lambda_M)
=\Big\{f\in C^\infty(\T)\colon f(x)=\sum_{m=-M}^M a_m e^{2\pi imx}, \ a_m\in\C\Big\}. 
\]
We say the discrete set $S=\{s_j\}_{j=1}^J$ satisfies the $\Lambda_M$-\emph{minimum separation condition} if
\begin{equation}
	\label{min}
	\min_{1\leq j<k\leq J} |s_j-s_k|\geq \frac{2}{M},
\end{equation}
where $|\cdot|$ is the distance on $\T$. Let $M(\T,\Lambda_M)$ be the set of discrete measures on $\T$ whose support satisfies the $\Lambda_M$-minimal separation condition.

Cand\`{e}s and Fernandez-Granda \cite[Theorem 1.2]{candes2014towards} proved that if $\mu_0\in M(\T,\Lambda_M)$ for $M\geq 128$, then $\mu_0$ is the unique solution to the \emph{super-resolution problem},
\begin{equation}
\label{SR}
\tag{SR}
\inf \|\mu\|  
\quad \text{subject to}\quad
\mu\in M(\T)
\quad\text{and}\quad
P_M\mu=P_M \mu_0. 
\end{equation}
Their proof requires the assumption that $M\geq 128$ and it is unknown whether the theorem holds for all values of $M>0$. Further, the conclusion of their theorem still holds if we replace the numerical constant 2 in (\ref{min}) with a smaller constant and impose a stronger condition on $M$. For example, the conclusion holds if the 2 is replaced with 1.26 provided that $M\geq 10^3$ \cite[Theorem 2.2]{fernandez2016super}. 

As previously mentioned, we are concerned with the noisy case. For this model, instead of observing the noiseless data $P_M\mu_0$, suppose we are given the corrupted data, $P_M (\mu_0+\eta)$. The papers \cite{candes2013super,fernandez2013support} obtained an approximation of $\mu_0$ by solving the constrained minimization problem,
\begin{equation}
\label{SR1}
\tag{SR$_{\delta}$}
\inf \|\mu\|  
\quad \text{subject to}\quad
\mu\in M(\T)
\quad\text{and}\quad
\|P_M(\mu-\mu_0-\eta)\|_{L^2}\leq \delta,
\end{equation}
where $\delta>0$ can be freely chosen. On the other hand, the papers \cite{bhaskar2013atomic, tang2015near, azais2015spike, duval2015exact} studied the closely related unconstrained minimization problem,
\begin{equation}
\label{SR2}
\tag{SR$_{\tau}$}
\inf \(\frac{1}{2}\|P_M(\mu-\mu_0-\eta)\|_{L^2}^2+\tau\|\mu\|\)
\quad\text{subject to}\quad
\mu\in M(\T),
\end{equation} 
where $\tau>0$ can also be freely chosen. This problem is a special case of Tikhonov regularization.

Using standard weak-$\ast$ compactness arguments, it is not difficult to show that Problems (\ref{SR}), (\ref{SR1}), and (\ref{SR2}) are well-posed, i.e., the infimum in the three minimization problems can be replaced with the minimum. Further, appropriate dual formulations of all three problems can be recast as semi-definite programs, see \cite{candes2014towards,candes2013super,tang2015near}. 

The most important question in the study of regularization methods is to determine if the regularized solutions approximate the noiseless solution in some suitable sense. Suppose $\mu_\delta$ and $\mu_\tau$ are solutions to Problems (\ref{SR1}) and (\ref{SR2}), respectively. Intuitively speaking, we expect that $\mu_\delta$ and $\mu_\tau$ converge to $\mu_0$ if the parameters $\delta$ and $\tau$ are chosen appropriately depending on the noise level and the noise level tends to zero. This intuition is somewhat correct, since it is possible to show convergence for a subsequence and in the weak-$\ast$ sense.  

Such convergence statements are qualitative, whereas we want a quantitative bound. This leads us to the question: What is a natural way of quantifying the errors, $\mu_\delta-\mu_0$ and $\mu_\tau-\mu_0$? Burger-Osher \cite{burger2004convergence} argued that, since Tikhonov regularization is achieved in the weak-$\ast$ topology, it would be surprising if it is possible to bound the error in the total variation norm. Since Problem (\ref{SR2}) is a special case of Tikhonov regularization and is similar to Problem (\ref{SR1}), it is reasonable that the same principle applies. Numerical results have shown that the supports of $\mu_0$, $\mu_\tau$, and $\mu_\delta$ can be different \cite{candes2013super,duval2015exact}, which further supports this heuristic. Thus, it appears impossible to bound $\|\mu_\delta-\mu_0\|$ and $\|\mu_\tau-\mu_0\|$ in terms of the noise level. 

Since super-resolution is concerned with the recovery of fine details from coarse data, it is reasonable to bound $\mu_\delta-\mu_0$ and $\mu_\tau-\mu_0$ at small scales. Cand\`{e}s and Fernandez-Granda \cite{candes2013super} argued that it suffices to control smoothed-out errors at a certain resolution. For a kernel $K$, the smoothed out errors are $K*(\mu_\delta-\mu_0)$ and $K*(\mu_\tau-\mu_0)$.

% % % % % % % % % % % % % % % % % % % %
\subsection{Results}

We are primarily concerned with the solutions to Problems (\ref{SR1}) and (\ref{SR2}). In order to avoid addressing each method separately, we introduce the following definition. We say $\mu\in M(\T)$ is a \emph{$(\epsilon,\Lambda_M)$-approximation} of $\mu_0\in M(\T)$ if
\begin{equation}
	\|\mu\|\leq \|\mu_0\|+2\epsilon
	\quad\text{and}\quad
	\|P_M(\mu-\mu_0)\|_{L^2}\leq 2\epsilon. 
\end{equation}
The numerical constant 2 that appears in both inequalities is unimportant; our theorem still holds for any other sufficiently large constant. Propositions \ref{prop0} and \ref{prop1} show that solutions to either of the convex problems are $(\epsilon,\Lambda_M)$-approximations of $\mu_0\in M(\T;\Lambda_M)$, where $\epsilon$ depends on the noise. 

\begin{theorem}
	\label{thm1}
	There exists a constant $C>0$ such that the following hold. Suppose $\mu_0\in M(\T;\Lambda_M)$ for an integer $M\geq 128$ and $\mu$ is a $(\epsilon,\Lambda_M)$-approximation of $\mu_0$. For any twice differentiable $K$ with $K''\in L^\infty(\T)$, we have 
	\begin{equation}
	\label{eq0}
	\|K*(\mu-\mu_0)\|_{L^\infty}
	\leq C\epsilon \big(\|K\|_{L^\infty} + M^{-1}\|K'\|_{L^\infty} + M^{-2}\|K''\|_{L^\infty}\big). 
	\end{equation}
\end{theorem}

\begin{remark}
	\label{remark1}
	Since we are given noisy observations of $\mu_0$ up to frequency $M$ (equivalently, at scale $1/M$) and super-resolution is concerned with the recovery of fine details, we are particularly interested in quantifying the error $\mu-\mu_0$ at scale $1/N$, for integers $N>M$. There are two natural avenues for a defining a kernel $K_N$ that corresponds to a function of scale $1/N$.
	\begin{enumerate}[(a)]
		\item 
		The first is in the Fourier domain. Let $K_N\in C(\T;\Lambda_N)$. Important examples include the Dirichlet and Fej\'{e}r kernels. By Bernstein's inequality for trigonometric polynomials, we have
		\[
		\|K_N''\|_{L^\infty}
		\leq N\|K_N'\|_{L^\infty}
		\leq N^2\|K_N\|_{L^\infty}. 
		\]
		Inserting this into (\ref{eq0}), we obtain
		\[
		\|K_N*(\mu-\mu_0)\|_{L^\infty}
		\leq C\|K_N\|_{L^\infty} \(\frac{N}{M}\)^2 \epsilon. 
		\]
		
		\item
		The second is in the spatial domain. Suppose $k$ is twice differentiable, $k''$ is bounded, and $k$ is compactly supported in $[-\frac{1}{L},\frac{1}{L}]$ for some $L>2$. For an integer $N>M$, the function $k_N(x)=k(Nx)$ is compactly supported in $[-\frac{1}{LN},\frac{1}{LN}]$. Let $K_N\in C^2(\T)$ be the 1-periodization of $k_N$. We have
		\[
		\|K_N\|_{L^\infty}
		=\|k\|_{L^\infty},
		\quad
		\|K_N'\|_{L^\infty}
		\leq N\|k'\|_{L^\infty}
		\quad\text{and}\quad
		\|K_N''\|
		\leq N^2\|k''\|_{L^\infty}.
		\]
		Inserting this into (\ref{eq0}), we obtain
		\[
		\|K_N*(\mu-\mu_0)\|_{L^\infty}
		\leq C \max\big(\|k\|_{L^\infty},\|k'\|_{L^\infty},\|k''\|_{L^\infty}\big) \(\frac{N}{M}\)^2 \epsilon. 
		\]
	\end{enumerate}

\end{remark}

%\begin{remark}
%	There are several immediate consequences of the theorem. 
%	\begin{enumerate}[(a)]
%		
%		\item 
%		Since $\T$ has unit Lebesgue measure, we can deduce $L^p(\T)$ estimates for any $p\in(0,\infty]$. Indeed, there exists a universal constant $C>0$ (the same one as in the theorem) such that for all $p\in (0,\infty]$, we have
%		\[
%		\|K_N*\nu\|_{L^p}
%		\leq CC_0\(\frac{N}{M}\)^2\epsilon. 
%		\]
%		
%		\item
%		Since super-resolution is a spectral extrapolation problem, it might be of interest to bound $|\hat\nu(N)|$, as opposed to the averaging considered in the theorem. To do so, we specialize our results to the kernels, $K_N(x)=\sin(2\pi Nx)$ and $K_N(x)=\cos(2\pi Nx)$, to obtain the inequalities,
%		\begin{align*}
%		\frac{1}{2}\Big|\hat\nu(N)e^{2\pi iNx}-\hat\nu(-N)e^{-2\pi iNx}\Big|
%		&\leq CC_0\(\frac{N}{M}\)^2\epsilon, \\
%		\frac{1}{2}\Big|\hat\nu(N)e^{2\pi iNx}+\hat\nu(-N)e^{-2\pi iNx}\Big|
%		&\leq CC_0\(\frac{N}{M}\)^2\epsilon.
%		\end{align*}
%		They imply
%		\[
%		|\hat\nu(N)|
%		\leq CC_0\(\frac{N}{M}\)^2\epsilon,
%		\]
%		which shows that the error grows at the higher frequencies. 
%	\end{enumerate}
%\end{remark}

% % % % % % % % % % % % % % % % % % % %
\subsection{Related work}
The papers \cite{candes2013super, fernandez2013support, bhaskar2013atomic, tang2015near, azais2015spike, duval2015exact} assume $\mu_0$ is a discrete measure whose support satisfies the $\Lambda_M$-minimum separation condition and analyze either Problem (\ref{SR1}) or (\ref{SR2}).

Our result is completely different from the results contained in the aforementioned papers, with the exception of Cand\`{e}s and Fernandez-Granda \cite[Theorem 1.2]{candes2013super}. There are also some important differences between our Theorem \ref{thm1} and their theorem. 
\begin{enumerate}[(a)]
	
	\item 
	An important difference is that our result applies to both Problems (\ref{SR1}) and (\ref{SR2}), whereas their theorem only applies to the former de-noising method. To our best knowledge, we are the first to establish estimate (\ref{eq0}) for the latter method. 
	
	\item
	Further, their theorem requires weaker assumptions on the kernel and they obtain $L^1(\T)$ estimates. We require slightly stronger assumptions on the kernel, but in return, we obtain $L^\infty(\T)$ estimates and a greatly simplified proof. In fact, they use a complicated comparison of scales argument to derive their inequality, whereas we shall not require this type of argument. Importantly, our stronger assumptions on the kernel do not preclude any important cases, see Remark \ref{remark1}, and from this perspective, these assumptions come for free. 
\end{enumerate}

% % % % % % % % % % % % % % % % % % % %
\section{Proofs}
\label{sec proofs}

% % % % % % % % % % % % % % % % % % % %
\subsection{Notation}

Before we prove the theorem, we need to introduce some notation. For a discrete set $S=\{s_j\}_{j=1}^J\subset\T$ and integer $M>0$, let
\[
S_M(j)
=\{x\in\T\colon |x-s_j|\leq 0.16 M^{-1}\}.
\]
If $S$ satisfies the $\Lambda_M$-minimum separation condition and $j\not=k$, then $S_M(j)$ and $S_M(k)$ are disjoint. The constant $0.16$ was originally chosen by Cand\`{e}s and Fernandez-Granda \cite{candes2014towards,candes2013super} in order to somewhat minimize the constants that appeared in their arguments. The following results still hold if 0.16 is replaced with a smaller positive constant constant, but the constants that appear in Propositions \ref{prop2} and \ref{prop3} and Theorem \ref{thm1} would also change. For convenience, let  
\[
S_M=\bigcup_{j=1}^J S_M(j).
\]
 
For a vector $v\in\C^K$, let $v_k$ denote its $k$-th entry, and let $\|v\|_\infty=\max_{1\leq k\leq K}|v_k|$. For a $K\times K$ matrix $D$, let $\|D\|_\infty$ be its operator norm. Note that we reserve $\|\cdot\|_{L^\infty}$ for functions and $\|\cdot\|_\infty$ for vectors and matrices. 

Throughout the remainder of this paper, we shall write $A \lesssim B$ if there exists a universal constant $C>0$ such that $A\leq CB$. In particular, the constant $C$ is independent of $\mu,\mu_0,K,M,J,\delta,\tau,\epsilon$. 

% % % % % % % % % % % % % % % % % % % %
\subsection{Preliminary results}

The following proposition establishes the connection between $(\epsilon,\Lambda_M)$-approximations of $\mu_0$ and the solutions to Problems (\ref{SR1}) and (\ref{SR2}) under a certain noise model. The following result holds without assuming $\mu_0\in M(\T;\Lambda_M)$ or $M\geq 128$, and clearly generalizes to higher dimensions.

\begin{proposition}
	\label{prop0}
	Let $\mu_0\in M(\T)$ and $\eta\in L^2(\T)$ be unknown. Suppose we are given $P_M(\mu_0+\eta)$ for some integer $M>0$ and given $\epsilon>0$ such that $\|P_M\eta\|_{L^2}\leq\epsilon$. Set $\delta=\tau=\epsilon$. Then, any solution to Problem (\ref{SR1}) or (\ref{SR2}) is a $(\epsilon,\Lambda_M)$-approximation of $\mu_0$. 
\end{proposition}

\begin{proof}
	\indent
	\begin{enumerate}[(a)]
		\item 
		Let $\mu_\delta$ be a solution to Problem (\ref{SR1}). Observe that $\mu_0$ satisfies the constraint in Problem (\ref{SR1}) since
		\[
		\|P_M(\mu_0-\mu_0-\eta)\|_{L^2}
		=\|P_M\eta\|_{L^2}
		\leq\epsilon.
		\]
		By definition of $\mu_\delta$ being a solution, we have $\|\mu_\delta\|\leq \|\mu_0\|$. We also have
		\[
		\|P_M(\mu_\delta-\mu_0)\|_{L^2}
		\leq \|P_M(\mu_\delta-\mu_0-\eta)\|_{L^2}+\|P_M\eta\|_{L^2}
		\leq 2\epsilon. 
		\]
		\item
		Let $\mu_\tau$ be a solution to Problem (\ref{SR2}). By definition of $\mu_\tau$ being a solution, we have
		\[
		\epsilon \|\mu_\tau\|
		\leq \frac{1}{2} \|P_M(\mu_0-\mu_0-\eta)\|_{L^2}^2 + \epsilon \|\mu_0\|.
		\]
		Rearranging, we obtain $\|\mu_\tau\|\leq \|\mu_0\|+\epsilon/2$. The inequality,
		\[
		\|P_M(\mu_\tau-\mu_0)\|_{L^\infty}\leq \tau,
		\]
		requires more work and we refer to \cite[Lemma 1]{bhaskar2013atomic} for a proof. 
	\end{enumerate}
\end{proof}

The previous proposition assumed that the noise satisfies $\|P_M\eta\|_{L^2}\leq\epsilon$, and in particular, this implies $|\hat\eta(m)|\leq\epsilon$ for all $m\in\Lambda_M$. If we do not want to assume that $\hat\eta(m)$ is bounded, an alternative noise model is to assume that $\hat\eta(m)$ is a Gaussian random variable. The following proposition shows that, with high probability, solutions to both convex problems are still $(\epsilon,\Lambda_M)$-approximations. 

\begin{proposition}
	\label{prop1}
	Let $\mu_0\in M(\T)$ and $\eta\in L^2(\T)$ be unknown. Suppose we are given $P_M(\mu_0+\eta)$ for some integer $M>0$ and the real and complex parts of $\hat\eta(m)$ are i.i.d. Gaussian random variables with mean zero and variance $\sigma^2$. For a parameter $\gamma>0$, set 
	\begin{equation}
		\label{eq9}
		\epsilon
		=\delta
		=\tau
		=\sigma(1+\gamma)\sqrt{2(2M+1)}.
	\end{equation}
	With probability at least $1-e^{-2(2M+1)\gamma^2}$, any solution to Problem (\ref{SR1}) or (\ref{SR2}) is a $(\epsilon,\Lambda_M)$-approximation of $\mu_0$. 
\end{proposition}

\begin{proof}
	By Parseval's equality, note that
	\[
	\frac{1}{\sigma^2}\|P_M\eta\|_{L^2}^2
	=\sum_{m=-M}^M \frac{|\hat\eta(m)|^2}{\sigma^2},
	\]
	is a $\chi^2$ random variable with $2(2M+1)$ degrees of freedom. By inequality (4.3) in \cite[Section 4]{laurent2000adaptive}, for all $x>0$,
	\[
	\mathbb{P}\(\ \frac{1}{\sigma^2}\|P_M\eta\|_{L^2}^2
	\geq 2(2M+1)+2\sqrt{2(2M+1)x}+2x\)
	\leq e^{-x}. 
	\]
	Set $x=2(2M+1)\gamma^2$. Then,
	\[
	\mathbb{P}\(\|P_M\eta\|_{L^2}
	\geq \sigma(1+\gamma)\sqrt{2(2M+1)}\)
	\leq e^{-2(2M+1)\gamma^2}. 
	\]
	With probability at least $1-e^{-2(2M+1)\gamma^2}$, we have 
	\[
	\|P_M\eta\|_{L^2}
	\leq \sigma(1+\gamma)\sqrt{2(2M+1)}
	=\epsilon.
	\]
	The conclusion follows from Proposition \ref{prop0}. 
\end{proof}

The following proposition shows that a weighted integral of $|\mu-\mu_0|$ on $S_M^c$ can be controlled in terms of $\epsilon$, provided that the assumptions of Theorem \ref{thm1} hold. This result first appeared in \cite[Lemma 2.1]{candes2013super}, but only for the difference $|\mu_\delta-\mu_0|$. A similar, but not identical, result for $|\mu_\tau-\mu_0|$ was proved in \cite[Lemma 2]{tang2015near}. 

\begin{proposition}
	\label{prop2}
	There exists a constant $C>0$ such that the following hold. Suppose $\mu_0\in M(\T;\Lambda_M)$ for an integer $M\geq 128$, $S=\{s_j\}_{j=1}^J$ is the support of $\mu_0$, and $\mu$ is a $(\epsilon,\Lambda_M)$-approximation of $\mu_0$. Then,
	\begin{align*}
	\int_{S_M^c} \ d|\mu-\mu_0| &\leq C\epsilon, \\
	\sum_j\int_{S_M(j)} |x-s_j|^2\ d|\mu-\mu_0|(x) &\leq C M^{-2}\epsilon.
	\end{align*} 
\end{proposition}

\begin{proof}
	Let $\nu=\mu-\mu_0$. It was shown in \cite[Lemma 2.4]{candes2013super} that there exist $f\in C(\T;\Lambda_M)$ with $\|f\|_{L^\infty}\leq 1$ and universal constants $C_1,C_2>0$ such that 
	\begin{align*}
	\int_S\ d|\nu|
	&=\Big|\int_S f\ d\nu \Big| \\
	&\leq \Big|\int_{\T} f\ d\nu \Big| + \Big|\int_{S_M^c} f\ d\nu\Big| + \Big|\sum_j \int_{S_M(j)\setminus\{s_j\}} f\ d\nu \Big| \\
	&\leq \Big|\int_{\T} f\ d\nu \Big| +\int_{S^c}\ d|\nu| - C_1\int_{S_M^c}\ d|\nu| - C_2M^2 \int_{S_M} |x-s_j|^2\ d|\nu|(x). 
	\end{align*} 
	Rearranging, we obtain 
	\begin{equation}
		\label{eq6}
		C_1\int_{S_M^c}\ d|\nu| + C_2M^2 \int_{S_M} |x-s_j|^2\ d|\nu|(x)
		\leq \Big|\int_{\T} f\ d\nu \Big| + \int_{S^c}\ d|\nu| - \int_S\ d|\nu|. 
	\end{equation}
	By definition of $(\epsilon,\Lambda_M)$-approximation, $f\in C(\T;\Lambda_M)$, and that $\|f\|_{L^2}\leq \|f\|_{L^\infty}\leq 1$, we see that
	\begin{equation}
		\label{eq7}
		\Big|\int_{\T} f\ d\nu \Big|
		\leq \|f\|_{L^2}\|P_M\nu\|_{L^2}
		\leq 2\epsilon. 
	\end{equation}
	By definition of $(\epsilon,\Lambda_M)$-approximation and that $\mu_0$ is supported in $S$, we have
	\[
	2 \epsilon +\|\mu_0\|
	\geq \|\mu\|
	=\|\mu_0+\nu\|
	\geq \int_S\ d|\mu_0|-\int_S\ d|\nu| + \int_{S^c}\ d|\nu|.
	\]
	Rearranging this inequality, we obtain
	\begin{equation}
		\label{eq8}
		\int_{S^c}\ d|\nu| - \int_S\ d|\nu|
		\leq 2\epsilon. 
	\end{equation}
	Combining inequalities (\ref{eq6}), (\ref{eq7}) and (\ref{eq8}) completes the proof.
\end{proof}

The following proposition is a generalization of \cite[Lemmas 2.5 and 2.7]{candes2013super}, and shows that there exists $f\in C(\T;\Lambda_M)$ that behaves like an affine function on each $S_M(j)$. 

\begin{proposition}
	\label{prop3}
	There exists a constant $C>0$ such that the following hold. Suppose $M\geq 128$ and the set $S=\{s_j\}_{j=1}^J\subset\T$ satisfies the $\Lambda_M$-minimum separation condition. For any $a,b\in \C^J$, there exists $f\in C(\T^d;\Lambda_M)$ such that 	
	\begin{align*}
	\|f\|_{L^\infty}
	&\leq C (\|a\|_\infty + M^{-1} \|b\|_\infty), \\
	|f(x)-a_j-b_j(x-s_j)|
	&\leq C(M^2\|a\|_\infty + M \|b\|_\infty) |x-s_j|^2, \quad x\in S_M(j). 
	\end{align*}
\end{proposition}

\begin{proof}
	Following the recipe given in \cite[Section 2]{candes2014towards}, it is possible to explicitly construct the desired $f$. Let 
	\[
	G(x)=\(\frac{\sin ((\frac{M}{2}+1)\pi x)}{(\frac{M}{2}+1)\sin(\pi x)}\)^4,
	\]
	and note that $G\in C(\T;\Lambda_M)$. We claim that there exist $\alpha,\beta\in\C^J$ such that if we define $f$ by
	\[
	f(x)=\sum_j \alpha_j G(x-s_j) + \sum_j \beta_j G'(x-s_j),
	\]
	then 
	\begin{equation}
	\label{f3}
	f(s_j)=a_j, 
	\quad\text{and}\quad
	f'(s_j)=b_j.
	\end{equation}
	To see why, we define the matrices $D_0,D_1,D_2\in \C^{J\times J}$, where
	\[
	(D_0)_{j,k}=G(s_j-s_k),\quad
	(D_1)_{j,k}=G'(s_j-s_k), 
	\quad\text{and}\quad
	(D_2)_{j,k}=G''(s_j-s_k).
	\]
	To prove the existence of the desired $f$, it suffices to show that there exists a solution to system of equations,
	\[
	\begin{pmatrix}
	D_0 &D_1 \\ D_1 &D_2
	\end{pmatrix}
	\begin{pmatrix}
	\alpha \\ \beta
	\end{pmatrix}
	=
	\begin{pmatrix}
	a \\ b
	\end{pmatrix}.
	\]
	It was shown in \cite[Section 2]{candes2014towards} that the $\Lambda_M$-minimum separation condition on $S$ and the assumption $M\geq 128$ imply that the system is invertible and that the unique solution is given by
	\begin{align*}
	\alpha &= D_0^{-1}(a-D_1\beta), \\
	\beta &= (D_2-D_1D_0^{-1}D_1)^{-1} (b-D_1D_0^{-1}a). 
	\end{align*}
	This proves the existence of $f$ satisfying conditions (\ref{f3}).
	
	Next, we obtain estimates on $\alpha,\beta$. It was also shown in \cite[Section 2]{candes2014towards} that
	\begin{align*}
	\|D_0^{-1}\|_\infty &\lesssim 1,\\
	\|D_1\|_\infty &\lesssim M, \\
	\|(D_2-D_1D_0^{-1}D_1)^{-1}\|_\infty &\lesssim M^{-2}. 
	\end{align*}
	These inequalities imply
	\begin{align*}
	\|\beta\|_\infty 
	&\lesssim M^{-1}\|a\|_\infty + M^{-2} \|b\|_\infty, \\
	\|\alpha\|_\infty 
	&\lesssim \|a\|_\infty + M^{-1}\|b\|_\infty. 
	\end{align*}
	
	It was shown in \cite[Section 2]{candes2014towards} that
	\[
	\sum_{k\not=j} |G^{(\ell)}(x-s_k)|\lesssim M^\ell, \quad
	x\in S_M(j) \quad \text{and}\quad \ell=0,1,2,3.
	\]
	Since $G^{(\ell)}$ decays rapidly away from the origin, the above inequalities imply, for all $x\in\T$,
	\begin{align*}
	|f(x)|
	&\leq \|\alpha\|_\infty \sum_j |G(x-s_j)| + \|\beta\|_\infty \sum_j |G'(x-s_j)| \\
	&\lesssim \|a\|_\infty + M^{-1}\|b\|_\infty. 
	\end{align*}
	This proves the first inequality of the proposition. 
	
	On each $S_M(j)$, define the function $h_j(x)=f(x)-a_j-b_j(x-s_j)$.
	It follows from (\ref{f3}) that $h_j(s_j)=h_j'(s_j)=0$. For all $x\in S_M(j)$, we have 
	\begin{align*}
	|h_j''(x)|
	=|f''(x)| 
	&\leq \|\alpha\|_\infty \sum_k |G''(x-s_k)| +\|\beta\|_\infty  \sum_k |G'''(x-s_k)| \\
	&\lesssim M^2\|a\|_\infty + M \|b\|_\infty. 
	\end{align*}
	Using Taylor expansions of $h_j$ around $s_j$, we obtain 
	\[
	|f(x)-a_j-b_j(x-s_j)|
	\lesssim (M^2\|a\|_\infty + M \|b\|_\infty) |x-s_j|^2, \quad x\in S_M(j). 
	\]		
\end{proof}

% % % % % % % % % % % % % % % % % % % %
\subsection{Proof of Theorem \ref{thm1}} 

Let $\nu=\mu-\mu_0$ and fix $x_0\in\T$. Since $\mu_0\in M(\T;\Lambda_M)$, we know that $\mu_0$ is supported in some discrete set $S=\{s_j\}_{j=1}^J$ satisfying the $\Lambda_M$-separation condition. We have
\begin{align}
	\label{eq1}  
	\begin{split}
	|(K*\nu)(x_0)|
	&=\Big|\int_{\T} K(x_0-x)\ d\nu(x)\Big| \\
	&\leq \Big|\sum_j \int_{S_M(j)} K(x_0-x) \ d\nu(x)\Big| + \|K\|_{L^\infty}\int_{S_M^c} \ d|\nu|. 
	\end{split}	
\end{align}
The first-order Taylor expansion of $K(x_0-x)$ around the point $x_0-s_j$ on the interval $S_M(j)$ is
\[
K(x_0-x)
=K(x_0-s_j)+K'(x_0-s_j)(s_j-x)+ \frac{1}{2} K''(z_j)|x-s_j|^2, 
\quad x\in S_M(j),
\]
for some $z_j\in\T$ depending on $x_0,x,s_j$. Inserting this into (\ref{eq1}), we obtain
\begin{align}
	\label{eq2}
	\begin{split}
	|(K*\nu)(x_0)|
	&\leq \Big|\sum_j \int_{S_M(j)} (K(x_0-s_j) - K'(x_0-s_j)(x-s_j)) \ d\nu(x)\Big| \\
	&\quad + \|K''\|_{L^\infty} \sum_j \int_{S_M(j)} |x-s_j|^2\ d|\nu|(x) + \|K\|_{L^\infty}\int_{S_M^c} \ d|\nu|.
	\end{split}
\end{align}

To bound the first term on the right hand side, we use an interpolation argument. Let $a,b\in\C^J$ such that $a_j = K(x_0-s_j)$ and $b_j = -K'(x_0-s_j)$. Let $f\in C(\T;\Lambda_M)$ be a function satisfying the properties in Proposition \ref{prop3}. We have
\begin{align}
	\label{f1}
	\|f\|_{L^\infty} 
	&\lesssim \|K\|_{L^\infty} + M^{-1}\|K'\|_{L^\infty}, \\
	\label{f2}
	|f(x)-a_j-b_j(x-s_j)|
	&\lesssim (M^2\|K\|_{L^\infty} + M\|K'\|_{L^\infty}) |x-s_j|^2, \quad x\in S_M(j). 
\end{align}
Inequality (\ref{f2}) implies
\begin{align}
	\label{eq3}
	\begin{split}
	&\Big|\sum_j \int_{S_M(j)} (K(x_0-s_j) - K'(x_0-s_j)(x-s_j)) \ d\nu(x)\Big| \\
	&\quad \leq \Big|\sum_j \int_{S_M(j)} (f(x)-K(x_0-s_j)+K'(x_0-s_j)(x-s_j)) \ d\nu(x)\Big| +\Big| \int_{S_M} f \ d\nu\Big| \\
	&\quad \lesssim \big(M^2\|K\|_{L^\infty} + M\|K'\|_{L^\infty}\big) \sum_j \int_{S_M(j)} |x-s_j|^2 \ d|\nu|(x) + \Big|\int_{\T} f\ d\nu\Big| + \Big|\int_{S_M^c} f\ d\nu\Big|.
	\end{split}
\end{align}
Using inequality (\ref{f1}), we obtain
\begin{equation}
	\label{eq4}
	\Big|\int_{S_M^c} f\ d\nu\Big|
	\lesssim \big(\|K\|_{L^\infty} + M^{-1}\|K'\|_{L^\infty}\big) \int_{S_M^c} \ d|\nu|. 
\end{equation}
Using inequality (\ref{f1}) and the definition of a $(\epsilon,\Lambda_M)$-approximation, we see that
\begin{equation}
	\label{eq5}
	\Big|\int_{\T} f\ d\nu\Big|
	\leq \|f\|_{L^2}\|P_M\nu\|_{L^2}
	\lesssim \big(\|K\|_{L^\infty} + M^{-1}\|K'\|_{L^\infty}\big)\epsilon. 
\end{equation}
Combining inequalities (\ref{eq2}), (\ref{eq3}), (\ref{eq4}) and (\ref{eq5}), we obtain
\begin{align*}
|(K*\nu)(x_0)|
&\lesssim \big(\|K\|_{L^\infty} + M^{-1}\|K'\|_{L^\infty} \big)\epsilon \\
&\quad + \big(\|K\|_{L^\infty} + M^{-1}\|K'\|_{L^\infty} \big)\int_{S_M^c} \ d|\nu| \\
&\quad + \big(M^2\|K\|_{L^\infty}+M\|K'\|_{L^\infty}+\|K''\|_{L^\infty} \big) \sum_j \int_{S_M(j)} |x-s_j|^2\ d|\nu|(x).
\end{align*}
Finally, we apply Proposition \ref{prop2} to complete the proof.

% % % % % % % % % % % % % % % % % % % %
\section{Acknowledgements} 

The author thanks Professors John J. Bendetto and Hans G. Feichtinger for their helpful feedback and interest in the manuscript. This work was partially supported by DTRA Grant 1-13-1-0015.

% % % % % % % % % % % % % % % % % % % %
\nocite{}
\nocite{benedetto2010integration}
\bibliography{NoiseErrorbib}
\bibliographystyle{alpha}

% % % % % % % % % % % % % % % % % % % %
\end{document}